\documentclass[a4paper,10pt]{article}

\usepackage[bottom=30mm,top=27mm,left=19mm,right=19mm]{geometry}

\usepackage{graphicx,epstopdf}
\usepackage{amsmath,amsthm,amsfonts,amssymb,enumerate,bm,array,verbatim,multirow}
\usepackage[colorlinks=false]{hyperref}
\usepackage[titletoc,title]{appendix}

\newcommand{\eq}{\begin{equation}}
\newcommand{\en}{\end{equation}}
\newcommand{\eqa}{\begin{eqnarray}}
\newcommand{\ena}{\end{eqnarray}}
\newtheorem{theorem}{Theorem}

\begin{document}

\begin{center}
{\large{\bf  Quantum partial search for uneven distribution of multiple target items}}\\[7mm]

{\small Kun Zhang \footnote{kun.h.zhang@stonybrook.edu} and Vladimir Korepin
\footnote{vladimir.korepin@stonybrook.edu}}\\[3mm]

\small{\textit{$^{1}$Department of Physics and Astronomy, State University of New York at Stony Brook, Stony Brook, NY 11794-3800 \\
$^{2}$C.N. Yang Institute for Theoretical Physics, State University of New York at Stony Brook, Stony Brook, NY 11794-3840}} \\[1cm]

\end{center}

\begin{center}\parbox{13cm}{

\centerline{\small  \bf Abstract}

Quantum partial search algorithm is approximate search. It aims to find a target block (which has the target items). It runs a little faster than full Grover search. In this paper, we consider quantum partial search algorithm for multiple target items unevenly distributed in database (target blocks have different number of target items). The algorithm we describe can locate one of the target blocks. Efficiency of the algorithm is measured by number of queries to the oracle. We optimize the algorithm in order to improve efficiency. By perturbation method, we find that the algorithm runs the fastest when target items are evenly distributed in database.

\vspace{.5cm}

{\bf \small Key Words:} database search, approximate search, multiple targets, Grover search, quantum algorithm, optimization\\

{\bf \small  PACS numbers:} 03.67.-a, 03.67.Lx}

\end{center}
\vspace{.2cm}

\section{Introduction}

\label{intro}

Quantum partial search algorithm \cite{GR05,KG06,KX09} trades accuracy for speed, namely, the algorithm finds the block where the target item is located instead of the exact address of the target item. The algorithm is based on the full quantum search algorithm (famous Grover algorithm \cite{Grover96,Grover97}). Grover algorithm acquires quadratic speedup over corresponding classical search algorithm. Grover algorithm finds the target item in the database. Target item also called marked item in the literature. Grover algorithm is optimal \cite{Zalka99,BBHT96} and the quantum partial search algorithm is also optimal among other protocols \cite{Korepin05,Korepin06,KX07,KV06}. The generalization of partial search algorithm for multiple target items has been studied in \cite{CV07,CPB09}. In this paper, we consider the case that different blocks have different number of target items (uneven distribution).

The total number of items in the database is denoted by $N$. Classically, we consider database as a set of $N$ items. In quantum case, the database is an $N$ dimensional Hilbert space. Each basis vector $|y\rangle$ corresponds to an item. Both the Grover algorithm and partial search algorithm consist of repetitions of the Grover iterations. The Grover iteration is based on the oracle model \cite{NC11} and each of Grover iteration has one query to the oracle (black box). We use the number of queries to the oracle (number of Grover iteration) as the complexity measure. If there is only one target item, Grover algorithm finds the target item with probability close to 1 in
\eq
j_{\text{full}}=\frac \pi 4 \sqrt N+\mathcal O(\frac 1 N),\quad\quad N\rightarrow \infty
\en
number of queries to the oracle \cite{Grover96,NC11}. If we have total of $z$ target items, Grover algorithm can find one of the marked items in
\eq
j_{\text{full}}=\frac \pi 4 \sqrt{\frac N  z}+\mathcal O(\frac 1 N),\quad\quad N\rightarrow \infty
\en
queries to the oracle. Quantum partial search algorithm aims to find the blocks with target items instead of the accurate locations of the target items. We can divide the database of $N$ items into $K$ blocks which have $b$ items. Each block has the same number of items, i.e., $N=bK$. In general, partial search algorithm can win over the full Grover algorithm by a number scaling as $\sqrt b$ both in one target item case \cite{GR05,KG06} and multiple target case \cite{CV07,CPB09}. This is the result of optimization \cite{Korepin05,CV07}. We can run the partial search algorithm on different set of partitions simultaneously (those are partitions popular with users).

Number of target blocks is denoted by $t$. In this paper, we only consider $t<K/4$ (which is worthwhile for partial search algorithm). When $t\geq K/4$, we can randomly pick up one block and run the full quantum search algorithm for unknown number of marked items proposed in \cite{BBHT96}. For one block, we need at most $4\sqrt b$ number of queries to the oracle to find one of the marked items or to claim that there is no target items in this block (with trivial small probabilities of failure). If no marked items in this block was found, we pick up another block and run the search again. The total expected number of Grover iterations is therefore
\eq
j_{t\geq N/4}\leq\sum_{l=0}^{\infty}\left(\frac 3 4\right)^l4\sqrt b=16\sqrt b\sim\mathcal O(\sqrt b)
\en
However, when we do not know the number of target blocks, we may consider the partial search algorithm for unknown number of marked blocks and unknown number of marked items in future.

This paper has two parts. In first part, we study the partial search algorithm for uneven distribution of multiple target items. Then we optimize the algorithm in large block limit $b\rightarrow\infty$. In second part, we study the partial search algorithm for uneven distribution of multiple target items by perturbation method (perturbation from even distribution). We find that more queries is needed for uneven distribution in second order perturbation proportional to the variance of target distribution in blocks. For reader's convenience, a summary of notations is in Table \ref{table 1}.

\begin{table}[!htbp]
\caption{Summary of notations.}
\label{table 1}
\scriptsize
\begin{center}
  \setlength{\extrarowheight}{4pt}
  \begin{tabular}{c|c|c}\hline\hline
  Notation & Definition & Remarks \\ \hline
  $N$ & total number of items in the database & As $b\rightarrow\infty$, we have $N\rightarrow\infty$.\\ \hline
  $n$ & number of bits necessary to represent the database ($N$ items) & $n=\log_2 N$ \\ \hline
  $K$ & total number of blocks in the database & each block has same number of items\\ \hline
  $b$ & number of items in each block & $N=bK$, \quad\quad $b\rightarrow\infty$ \\ \hline
  $z$ & total number of target items in the database & target item also known as marked item\\ \hline
  $t$ & number of target blocks in the database & $t<K/4$\\ \hline
  $\beta$ & ratio between $t$ and $K$ & $\beta=t/K$ \\ \hline
  $\tau_i$ &  number of marked items in $i$-th marked block & $z=\sum_i \tau_i$,\quad\quad $i=1,2,\cdots t$ \\ \hline
  $\bar\tau$ & average number of marked items in a marked block & $\bar\tau=z/t$  \\ \hline
  $\varepsilon_i$ & $\tau_i$ is perturbation of $\bar\tau$ & $\tau_i=\bar\tau(1+\varepsilon_i)$, \quad\quad $\varepsilon_i$ is small \\ \hline
  $\delta^2(\tau)$ & variance of $\tau_i$ & $\delta^2(\tau)=\left(\bar\tau^2 /t\right)\sum_{i=1}^t \varepsilon_i^2$ \\ \hline
  $A$ & set of all target items in the database& number of elements in $A$ is $z$  \\ \hline
  $A_i$ & set of all target items in $i$-th target block & number of elements in $A_i$ is $\tau_i$  \\ \hline
  $X$ & set of all non-target items in the database & number of elements in $X$ is $(N-z)$  \\ \hline
  $X_i$ & set of all non-target items in $i$-th target block & number of elements in $X_i$ is $(b-\tau_i)$  \\ \hline
  $\hat G_1$ & global Grover iteration & refer to formulae (\ref{def G1})-(\ref{define I s1}) \\ \hline
  $\hat G_2$ & local Grover iteration & see equations (\ref{def G local})-(\ref{def s2})\\ \hline
  $j_1$ & number of global Grover iterations & $j_1\propto \sqrt N\rightarrow\infty$ \\ \hline
  $j_2$ & number of local Grover iterations & $j_2\propto \sqrt b\rightarrow\infty$\\ \hline
  $\hat I$ & identity operator & \\ \hline
  $|s_1\rangle$ & average of all items in the database & $|s_1\rangle=1/\sqrt N\sum_{y=0}^{N-1}|y\rangle$ \\ \hline
  $|s_2\rangle$ & average of all items in a block & $|s_2\rangle=1/\sqrt b\sum_{y\in\text{one block}}|y\rangle$ \\ \hline
  $\hat I_{s_1}$  & reflection about the state $|s_1\rangle$ & $\hat I_{s_1}=\hat I-2|s_1\rangle\langle s_1|$; see equation (\ref{define I s1})\\ \hline
  $\hat I_{s_2}$ & reflection about the state $|s_2\rangle$ & $ \hat I_{s_2}=\hat I-2|s_2\rangle\langle s_2|$; see equation (\ref{def I s2})\\ \hline
  $\hat I_T$ & reflection about all the target states & $\hat I_T=\hat I-2\sum_{t\in A}^z|t\rangle\langle t|$; see equation (\ref{define I T}) \\ \hline
  $\theta$ & angle of global Grover iteration & $\sin^2\theta=z/N$; see equation (\ref{def theta}) \\ \hline
  $\theta_i$ & angles of local Grover iterations in $i$-th target blocks & $\sin^2\theta_i=\tau_i/b$; see equation (\ref{def theta q})\\ \hline
  $|t_i\rangle$ & average of target items in $i$-th target blocks & $|t_i\rangle=1/\sqrt {\tau_i}\sum_{t\in A_i}^{\tau_i}|t\rangle$ \\ \hline
  $|ntt_i\rangle$ & average of non-target items in $i$-th target blocks & $|ntt_i\rangle=1/\sqrt{b-\tau_i}\sum_{y\in X_i}^{b-\tau_i}|y\rangle$ \\ \hline
  $|u\rangle$ & average of all items in non-target blocks & see equation (\ref{def nt}) \\ \hline
  $\eta$ & \multirow{2}{*}{defined by the formula on the right side} & $j_1=\frac \pi 4 \sqrt{N/z}-\eta\sqrt b$ \\ \cline{1-1}\cline{3-3}
  $\alpha$ &   &  $j_2=\alpha\sqrt b$ \\ \hline
  $\eta_0$ & optimal value of $\eta$ in even distribution case &  \multirow{2}{*}{see equation (\ref{eta0 alpha0})} \\ \cline{1-2}
  $\alpha_0$ & optimal value of $\alpha$ in even distribution case & \\ \hline
  $\eta_K$ & optimal value of $\eta$ in uneven distribution case & determined by cancellation equation (\ref{eqa cancellation limit different}) \\ \hline
  $\alpha_K$ & optimal value of $\alpha$ in uneven distribution case & determined by optimization condition (\ref{opti relation}) \\ \hline
  $f(\eta,\alpha)$ & function to maximize in order to minimize the number of iterations & $f=\eta-\alpha$
  \\ \hline\hline
  \end{tabular}
\end{center}
\end{table}

\section{Partial search algorithm for uneven distribution of target items}

\label{section algorithm}

In this section, we study the partial search algorithm of uneven distribution of target items in details. We use two kinds of Grover iterations: global and local Grover iterations. The number of global and local Grover iterations should satisfy certain constraint called cancellation equation. We study optimization of the algorithm under such constraint.

\subsection{Steps of the algorithm}

\label{subsection steps}

Assume that the total number of items in the database is a power of 2, i.e., $N=2^n$. We can construct the uniform superposition of all basis vectors fast and efficiently by applying the Hadamard gate $H$ \cite{NC11}:
\eq
\label{def s1}
|s_1\rangle=H^{\otimes n}|0\rangle=\frac 1 {\sqrt N}\sum_{y=0}^{N-1}|y\rangle,\quad\quad\quad \langle s_1| s_1\rangle=1
\en
Here $|y\rangle$ is an element of orthonormal basis. The database is represented by an $N$ dimensional Hilbert space. The steps of partial search algorithm are listed below:

\begin{enumerate}[\bfseries Step 1.]
  \item $j_1$ global Grover iterations $\hat G_1$ defined as
  \eq
  \label{def G1}
  \hat G_1=-\hat I_{s_1}\hat I_T
  \en
  The operator
  \eq
  \label{define I T}
  \hat I_T=\hat I-2\sum_{t\in A}^z|t\rangle\langle t|
  \en
  is a reflection in a plane perpendicular to all target items. Here $z$ is the total number of target items in the database; $A$ is the set of all target items and $\hat I$ is the identity operator.  The operator
  \eq
  \label{define I s1}
  \hat I_{s_1}=\hat I-2|s_1\rangle\langle s_1|
  \en
  is a reflection in a plane perpendicular to the average of items. Let us explain reflection in the average. For example, for an arbitrary vector $|v\rangle$
  \eq
  |v\rangle=\sum_{y=0}^{N-1}a_y|y\rangle
  \en
  Here $a_y$ is a complex number. The operator $-\hat I_{s_1}$ acts as
  \eq
  \label{eqa Is1 reflect}
  -\hat I_{s_1}|v\rangle =\sum^{N-1}_{y=0}\breve{a}_y|y\rangle,\quad\quad\quad \breve{a}_y=2\bar{a}-a_y,\quad\quad\quad \bar a=\sum_{y=0}^{N-1}\frac {a_y}{N}
  \en
  Reflection operators can also be viewed as rotations. The rotation angle of global Grover iteration $\hat G_1$ (\ref{def G1}) is
  \eq
  \label{def theta}
  \sin^2\theta=\frac{z}{N}
  \en

  \item $j_2$ local Grover iterations $\hat G_2$. The local iteration is defined by
  \eq
  \label{def G local}
  \hat G_2=-\left(\bigoplus_\text{blocks}^K\hat I_{s_2}\right)\hat I_{T}
  \en
  The local operator $\hat I_{s_2}$ is given by
  \eq
  \label{def I s2}
  \hat I_{s_2}=\hat I-2|s_2\rangle\langle s_2|
  \en
  Here
  \eq
  \label{def s2}
  |s_2\rangle=\frac 1 {\sqrt b} \sum_{\hspace{1mm}y\in \text{one block}}|y\rangle
  \en
  is average of all items in a block. Local Grover iteration is usual Grover iteration for a block (considered as a database). Therefore, one can run the local Grover iteration on the same hardware which is used for global Grover iteration. Direct sum means that we run local search in each block simultaneously. The rotation angle for local Grover iteration $\hat G_2$ (\ref{def G local}) is
  \eq
  \label{def theta q}
  \sin^2\theta_i=\frac{\tau_i}{b}
  \en

  \item One last reflection $\hat I_{s_1}$ (\ref{define I s1}) vanishes amplitudes of all items in non-target blocks.

  \item Measurement will reveal a target block with high probability.
\end{enumerate}

During the algorithm, in a target block, amplitudes of all target items are the same. So we can only follow the amplitude of the average of all target items in one block:
\eq
\label{def Tq}
|t_i\rangle=\frac{1}{\sqrt {\tau_i}}\sum_{t\in A_i}^{\tau_i}|t\rangle
\en
It is the normalized sum of all target items in $i$-th target block. The set of all target items in $i$-th target block is denoted as $A_i$. Also the amplitudes of non-target items in one target block are the same. Therefore we define the normalized sum of all non-target items in $i$-th target block:
\eq
\label{def ntt}
|ntt_i\rangle=\frac 1 {\sqrt{b-\tau_i}}\sum_{y\in X_i}^{b-\tau_i}|y\rangle
\en
Here the set of all non-target items in $i$-th target block is denoted as $X_i$. The normalized sum of items in all non-target block is denoted by $|u\rangle$:
\eq
\label{def nt}
|u\rangle=\frac 1 {\sqrt{N-bt}}\sum_{\hspace{1mm}y\in\text{non-target blocks} }^{N-bt}|y\rangle,
\en
where $y\in \left(X-\sum_iX_i\right)$ and $X$ is the set of all non-target items in the database. In conclusion, the algorithm is a representation of $SO(2t+1)$ group. Because of local Grover iteration $\hat G_2$ (\ref{def G local}) acting locally on the blocks, therefore, in the bases $|t_i\rangle$, $|ntt_i\rangle$ and $|u\rangle$, $\hat G_2$ has the block diagonal form:
      \eq
      \label{G2 matrix}
      \hat G^{j_2}_{2}=\left(
        \begin{array}{cccccccc}
          \cos(2j_2\theta_1) & \sin(2j_2\theta_1)  & 0 & 0 & \cdots  & 0 & 0 & 0 \\
          -\sin(2j_2\theta_1)  & \cos(2j_2\theta_1)  & 0 & 0 & \cdots  & 0 & 0 & 0 \\
          0 & 0 & \cos(2j_2\theta_2) & \sin(2j_2\theta_2) & \cdots  & 0 & 0 & 0 \\
          0 & 0 & -\sin(2j_2\theta_2)  & \cos(2j_2\theta_2) & \cdots  & 0 & 0 & 0 \\
          \vdots & \vdots & \vdots & \vdots & \ddots & \vdots & \vdots & \vdots \\
          0 & 0 & 0 & 0 & \cdots & \cos(2j_2\theta_t) & \sin(2j_2\theta_t) & 0 \\
          0 & 0 & 0 & 0 & \cdots & -\sin(2j_2\theta_t) & \cos(2j_2\theta_t) & 0 \\
          0 & 0 & 0 & 0 & \cdots & 0 & 0 & 1 \\
        \end{array}
      \right),
  \en
which is a $(2t+1)\times (2t+1)$ orthogonal matrix. Note that local iteration $\hat G_2$ acts trivially on non-target blocks.

\subsection{Cancellation equation}

\label{subsection cancellation}

After \textbf{Step 1} ($j_1$ global Grover iterations $\hat G_1$), we get the state of the database:
\eqa
\hat G_1^{j_1}|s_1\rangle&=&\frac{\sin\left(\left(2j_1+1\right)\theta\right)}{\sqrt{z}}   \sum_{t\in A}|t\rangle+\frac{\cos\left(\left(2j_1+1\right)\theta\right)}{\sqrt {N-z}}\sum_{y\in X }|y\rangle \\
&=& \frac{\sin\left(\left(2j_1+1\right)\theta\right)}{\sqrt{z}}\sum_{i=1}^t\sqrt{\tau_i}|t_i\rangle+\frac{\cos\left(\left(2j_1+1\right)\theta\right)}{\sqrt {N-z}}\sum_{i=1}^t\sqrt{b-\tau_i}|ntt_i\rangle\nonumber\\
&&+\cos\left(\left(2j_1+1\right)\theta\right)\sqrt{\frac{b(K-t)}{N-z}}|u\rangle\nonumber
\ena
Here we rewrite the result in the bases $|t_i\rangle$ (\ref{def Tq}), $|ntt_i\rangle$ (\ref{def ntt}) and $|u\rangle$ (\ref{def nt}). The amplitude of all the non-target block states is $a_{nt}$:
\eq
\label{def ant}
a_{nt}=\frac{\cos\left(\left(2j_1+1\right)\theta\right)}{\sqrt{N-z}}
\en
It will remain unchanged during the \textbf{Step 2}.

In order to calculate $\hat G_2^{j_2}\hat G_1^{j_1}|s_1\rangle$, we can directly apply the matrix formalism of $\hat G_2^{j_2}$ (\ref{G2 matrix}) on the state $\hat G_1^{j_1}|s_1\rangle$. Then we have
\eq
\label{eqa G2G1 state}
\hat G_2^{j_2}\hat G_1^{j_1}|s_1\rangle=\sum_{i=1}^ta_{t_i}|t_i\rangle+\sum_{i=1}^ta_{ntt_i}|ntt_i\rangle+a_{nt}\sqrt{(b(K-t))}|u\rangle
\en
with
\eqa
\label{a t i}a_{t_i} &=& \sqrt{\frac{\tau_i}{z}}\cos\left(2j_2\theta_i\right)\sin\left(\left(2j_1+1\right)\theta\right)
+\left(\sqrt{\frac{b-\tau_i}{N-z}}\right)\sin\left(2j_2\theta_i\right)\cos\left(\left(2j_1+1\right)\theta\right); \\
\label{a ntt i}a_{ntt_i} &=& -\sqrt{\frac{\tau_i}{z}}\sin\left(2j_2\theta_i\right)\sin\left(\left(2j_1+1\right)\theta\right)
+\left(\sqrt{\frac{b-\tau_i}{N-z}}\right)\cos\left(2j_2\theta_i\right)\cos\left(\left(2j_1+1\right)\theta\right)
\ena
The vector (\ref{eqa G2G1 state}) describes the state of the database after \textbf{Step 2}.

After \textbf{Step 3}, the amplitudes of states in non-target block should vanish. Specifically, the amplitude $a_{nt}$ should be twice of the average, i.e., $a_{nt}=2\bar a$, because the operator $-\hat I_{s_1}$ inverts the amplitudes about the average, see equation (\ref{eqa Is1 reflect}). Then we have the constraint relation
\eq
a_{nt}=\frac 2 N\left(b(K-t)a_{nt}+\sum_{i=1}^t\left(a_{t_i}\sqrt{\tau_i}+a_{ntt_i}\sqrt{b-\tau_i}\right)\right),
\en
which leads to
\eq
N\left(\frac t K-\frac 1 2\right)a_{nt}=\sum_{i=1}^t\left(a_{t_i}\sqrt{\tau_i}+a_{ntt_i}\sqrt{b-\tau_i}\right)
\en
Substituting $a_{nt}$ (\ref{def ant}), $a_{t_i}$ (\ref{a t i}) and $a_{ntt_i}$ (\ref{a ntt i}) into above relation, we have
\eqa
\label{eqa cancellation different}
&&\frac N {\sqrt{N-z}}\left(\frac t K-\frac 1 2\right)\cos\left((2j_1+1)\theta\right) \\
&=&\sum_{i=1}^t\left(\frac{\tau_i}{\sqrt z}\cos\left(2j_2\theta_i\right)\sin\left((2j_1+1)\theta\right)+\sqrt{\frac{\tau_i(b-\tau_i)}{N-z}}\sin\left(2j_2\theta_i\right)\cos\left((2j_1+1)\theta\right)\right.\nonumber\\
&&\left.-\sqrt{\frac{\tau_i(b-\tau_i)}{z}}\sin\left(2j_2\theta_i\right)\sin\left((2j_1+1)\theta\right)+\frac{b-\tau_i}{\sqrt{N-z}}\cos\left(2j_2\theta_i\right)\cos\left((2j_1+1)\theta\right)\right)\nonumber
\ena
which is called \textbf{cancellation equation}. If above relation holds, the amplitudes of items in non-target block will vanish. And the final state is
\eq
\label{eqa final state}
\hat I_{s_1}\hat G_2^{j_2}\hat G_1^{j_1}|s_1\rangle=\sum_{i=1}^t \left(\left(a_{t_i}-a_{nt}\sqrt{\tau_i}\right)|t_i\rangle+\left(a_{ntt_i}-a_{nt}\sqrt{b-\tau_i}\right)|ntt_i\rangle\right)
\en
So the measurement will reveal a target block.

\subsection{Large block limit: $b\rightarrow\infty$}

\label{subsection large block}

The numbers of iterations $j_1$ and $j_2$ usually scale as \cite{Korepin05,Korepin06}
\eq
\label{def eta alpha}
j_1=\frac \pi 4 \sqrt{\frac{N}{z}}-\eta\sqrt b,\quad\quad\quad j_2=\alpha\sqrt b; \quad\quad\quad\eta>0,\quad\quad\quad\alpha>0,
\en
when $N\rightarrow\infty$. In thermodynamics limit $b\rightarrow \infty$, the cancellation equation (\ref{eqa cancellation different}) reduces into
\eq
\label{eqa cancellation limit different before}
\left(\frac t {\sqrt K}-\frac{\sqrt K}{2}\right)\sin\left(2\eta\sqrt{\frac z K}\right)
=\sum_{i=1}^t\left(\frac 1 {\sqrt K}\cos\left(2\alpha\sqrt{\tau_i}\right)\sin\left(2\eta\sqrt{\frac z K}\right)-\sqrt{\frac{\tau_i}{z}}\sin\left(2\alpha\sqrt{\tau_i}\right)\cos\left(2\eta\sqrt{\frac z K}\right)\right),
\en
which has a simpler expression for $\eta$
\eq
\label{eqa cancellation limit different}
\tan\left(2\eta\sqrt{\frac z K}\right)=\frac{2\sqrt K\sum_{i=1}^t\sqrt{\tau_i}\sin\left(2\alpha\sqrt{\tau_i}\right)}{\sqrt z\left(K-4\sum_{i=1}^t\sin^2\left(\alpha\sqrt\tau_i\right)\right)}
\en
Note the denominator on RHS is always positive if $t<K/4$.

\subsection{Optimization}

\label{subsection optimization}

We use the number of queries to the oracle as complexity measure of the algorithm. In order to accelerate the algorithm, we have to optimize. We want to minimize total number of queries to the oracle given by $j_1+j_2$:
\eq
\label{def total queries}
j_1+j_2=\frac \pi 4 \sqrt{\frac N z}-(\eta-\alpha)\sqrt b
\en
Therefore we want to maximize the function $f(\eta,\alpha)=\eta-\alpha$. By Lagrange multiplier method, we construct the function
\eq
\label{Lagrangian}
\mathcal L(\eta,\alpha,\lambda)=f(\eta,\alpha)-\lambda\left(\sqrt z\tan\left(2\eta\sqrt{\frac z K}\right)\left(2\sum_{i=1}^t\cos\left(2\alpha\sqrt{\tau_i}\right)+K-2t\right)-2\sqrt K\sum_{i=1}^t \sqrt{\tau_i} \sin\left(2\alpha\sqrt{\tau_i}\right)       \right)
\en
Here $\lambda$ is the Lagrangian multiplier. Maximization of $\mathcal L(\eta,\alpha,\lambda)$ leads to the equations
\eqa
0&=& 1-\lambda\left(\frac{2z}{\sqrt K \cos^2\left(2\eta\sqrt{\frac z K}\right)}\left(2\sum_{i=1}^t\cos\left(2\sqrt{\tau_i}\alpha\right)+K-2t     \right)\right); \\
0&=& -1+\lambda\left(4\sqrt z \tan\left(2\eta\sqrt{\frac z K}\right)\sum_{i=1}^t \sqrt{\tau_i}\sin\left(2\sqrt{\tau_i}\alpha\right)+4\sqrt K\sum_{i=1}^t\tau_i\cos\left(2\sqrt{\tau_i}\alpha\right)        \right); \\
0&=& \sqrt z\tan\left(2\eta\sqrt{\frac z K}\right)\left(2\sum_{i=1}^t\cos\left(2\sqrt{\tau_i} \alpha\right)+K-2t\right)-2\sqrt K\sum_{i=1}^t \sqrt{\tau_i} \sin\left(2\sqrt{\tau_i} \alpha\right)
\ena
We combine these equations to eliminate Lagrangian multiplier $\lambda$ and $\eta$ and get one equation for $\alpha$:
\eq
\left(2\left(\sum_{i=1}^t\cos\left(2\alpha\sqrt{\tau_i}\right) \right)+K-2t\right)\left(2\left(\sum_{i=1}^t\left(K\tau_i-z\right)\cos\left(2\alpha\sqrt{\tau_i}\right) \right)-z(K-2t)   \right)=0
\en
Note that $2\left(\sum_{i=1}^t\cos\left(2\alpha\sqrt{\tau_i}\right) \right)+K-2t>K-4t$. Besides, we only consider the case $(K-4t)>0$. Therefore, only second factor can vanish
\eq
\label{opti relation}
2\left(\sum_{i=1}^t\left(K\tau_i-z\right)\cos\left(2\alpha\sqrt{\tau_i}\right) \right)-z(K-2t) =0
\en
for maximum value of $f(\eta,\alpha)$. This is \textbf{optimization condition}. The solutions are denoted as $\eta_K$ and $\alpha_K$ respectively. Let us compare to the trivial even distribution case. When $\tau_i=\bar\tau$, we retrieve the optimal value for least number of iterations in even distribution database case \cite{CV07}. We denote the optimal values of $\alpha$ and $\eta$ in even distribution as $\alpha_0$ and $\eta_0$:
\eq
\label{eta0 alpha0}
\tan\left(2\eta_0\sqrt{\frac{z}{K}}\right)=\frac{\sqrt{3tK-4t^2}}{K-2t},\quad\quad\quad \cos\left(2\alpha_0\sqrt{\bar\tau}\right)=\frac{K-2t}{2(K-t)}
\en

\section{Perturbation}

\label{section prove}

Numerical results show that the more uniform the distribution of target items, the less queries to the oracle are needed \cite{CPB09}. In this section, we consider the uneven distribution of target items as the perturbation from even distribution, namely
\eq
\label{def varepsilon}
\tau_i=\bar\tau(1+\varepsilon_i),\quad\quad\quad \sum_{i=1}^t\varepsilon_i=0
\en
Perturbation $\varepsilon_i$ is a small number in the range $-1<\varepsilon_i<1$. Note that the variance of $\tau_i$ is
\eq
\label{def variance}
\delta^2(\tau)=\frac{\bar\tau^2}{t}\sum_{i=1}^t \varepsilon_i^2
\en
As the result, the optimal values $\eta_K$ and $\alpha_K$ change as
\eq
\label{def nabla}
\eta_K=\eta_0+\Delta\eta_K,\quad\quad\quad \alpha_K=\alpha_0+\Delta\alpha_K
\en

\subsection{The limit of many blocks: $K\rightarrow\infty$}

\label{subsection many block}

\begin{theorem}
\label{theorem 1}
In the limit of many blocks ($K\rightarrow\infty$), \textbf{uneven distribution requires more queries to the oracle than even distribution by}
\eq
\label{eqa theorem 1}
f(\eta_0,\alpha_0)-f(\eta_K,\alpha_K)=\left(\frac{\sqrt 3 \pi^2-3\pi+9\sqrt 3}{144}\right)\frac{\delta^2(\tau)}{\bar\tau^{5/2}}\approx 0.1615\frac{\delta^2(\tau)}{\bar\tau^{5/2}}
\en
\textbf{in second order perturbation}.
\end{theorem}
\begin{proof}
We notice that $\Delta \eta_K$ and $\Delta \alpha_K$ are nonzero at least in the second order of $\varepsilon_i$, because
\eq
\sum_{i=1}^t\varepsilon_i=0
\en
In following calculations, we keep the second order of $\varepsilon_i$, first order of $\Delta \eta_K$ and $\Delta \alpha_K$. In the limit $K\rightarrow\infty$, cancellation equation (\ref{eqa cancellation limit different}) simplifies as
\eq
\label{large K cancellation}
\eta=\frac 1 {z}\sum_{i=1}^t\sqrt{\tau_i}\sin\left(2\alpha\sqrt{\tau_i}\right)
\en
In second order of $\varepsilon_i$, we have
\eq
\sqrt{\tau_i}\approx\sqrt{\bar\tau}\left(1+\frac 1 2 \varepsilon_i-\frac 1 8 \varepsilon_i^2\right),
\en
which gives rise to
\begin{multline}
\label{eqa sin expand}
\sum_{i=1}^t\sqrt{\tau_i}\sin\left(2\alpha_K\sqrt{\tau_i}\right)=t\sqrt{\bar\tau}\sin\left(2\alpha_0\sqrt{\bar\tau}\right)+2t\bar\tau\cos\left(2\alpha_0\sqrt{\bar\tau}\right)\Delta\alpha_K\\
-\left(\frac t {2\sqrt{\bar\tau}}\alpha_0^2\sin\left(2\alpha_0\sqrt{\bar\tau}\right)
-\frac{t}{4\bar\tau}\alpha_0\cos\left(2\alpha_0\sqrt{\bar\tau}\right)
+\frac{t}{8\bar\tau\sqrt{\bar\tau}}\sin\left(2\alpha_0\sqrt{\bar\tau}\right)\right)\delta^2(\tau)
\end{multline}
Substituting above relation into (\ref{large K cancellation}), we get
\begin{multline}\label{large K second order}
\eta_0+\Delta\eta_K=\frac 1 {\sqrt{\bar\tau}}\sin\left(2\alpha_0\sqrt{\bar\tau} \right)+2\cos\left(2\alpha_0\sqrt{\bar\tau} \right)\Delta\alpha_K\\
+\left(-\frac{\alpha_0^2}{2\bar\tau\sqrt{\bar\tau}}\sin\left(2\alpha_0\sqrt{\bar\tau} \right)+\frac {\alpha_0} {4\bar\tau^2}\cos\left(2\alpha_0\sqrt{\bar\tau} \right)-\frac 1 {8\bar\tau^2\sqrt{\bar\tau}}\sin\left(2\alpha_0\sqrt{\bar\tau} \right)\right)\delta^2(\tau)
\end{multline}
On the other hand, from (\ref{large K cancellation}), in the case of even distribution, we know
\eq
\eta_0=\frac 1 {\sqrt{\bar\tau}}\sin\left(2\alpha_0\sqrt{\bar\tau} \right)
\en
And in the limit $K\rightarrow\infty$, parameter $\alpha_0$ (\ref{eta0 alpha0}) is
\eq
\alpha_0=\frac \pi   {6\sqrt {\bar\tau}}
\en
with
\eq
\sin\left(2\alpha_0\sqrt{\bar\tau} \right)=\frac {\sqrt 3}{2}, \quad \quad\quad \cos\left(2\alpha_0\sqrt{\bar\tau} \right)=\frac 1 2
\en
Substituting above relations into (\ref{large K second order}), we get
\eq
\Delta\eta_K=\Delta\alpha_K-\left(\frac{\sqrt 3 \pi^2-3\pi+9\sqrt 3}{144}\right)\frac{\delta^2(\tau)}{\bar\tau^{5/2}}
\en
Remember that $f=\eta-\alpha$, then
\eq
\begin{split}
f(\eta_0,\alpha_0)-f(\eta_K,\alpha_K)&=(\eta_0-\alpha_0)-(\eta_K-\alpha_K)=-\Delta\eta_K+\Delta\alpha_K\\
&=\left(\frac{\sqrt 3 \pi^2-3\pi+9\sqrt 3}{144}\right)\frac{\delta^2(\tau)}{\bar\tau^{5/2}}\approx 0.1615\frac{\delta^2(\tau)}{\bar\tau^{5/2}}
\end{split}
\en
\textbf{So we proved Theorem \ref{theorem 1}}.
\end{proof}

\subsection{Finite number of blocks}

\label{subsection finte block}

\begin{theorem}
\label{theorem 2}
Consider finite number of blocks (number of blocks is $K$). In second order of perturbation around even distribution, \textbf{uneven distribution of target items requires more queries then even distribution by}
\eq
\label{theorem ineqa}
  f(\eta_0,\alpha_0)-f(\eta_K,\alpha_K)>g(\beta)\frac{\delta^2{(\tau)}}{\bar\tau^{5/2}}
  \en
  \textbf{with} $\beta=t/K$ \textbf{and}
  \eq
  \label{def g beta}
  g(\beta)=\frac{\sqrt{3-4\beta}(1-2\beta)(\pi^2(1-\beta)+9)+3\pi(-8\beta^2+7\beta-1)}{144(1-\beta)}
  \en
 When $0<\beta<\beta_c$, function $g(\beta)$ is positive. Here $\beta_c$ is a root of polynomial of $5$ degree: $g(\beta_c)=0$. Approximate expression is $\beta_c\approx0.6281$.
\end{theorem}
\begin{proof} The proof is in four steps.
\begin{enumerate}[\bfseries Step 1.]
  \item Find perturbation $\Delta\alpha_K$ (\ref{def nabla}) in first order of $\delta^2(\tau)$ (\ref{def variance}). Firstly, the optimization condition (\ref{opti relation}) for uneven distribution can be rewritten as
\eq
\label{opti uneven}
\sum_{i=1}^t2K\tau_i\cos\left(2\alpha_K\sqrt{\tau_i}\right)-2t\bar\tau\sum_{i=1}^t\cos\left(2\alpha_K\sqrt{\tau_i}\right)-z(K-2t)=0
\en
In the case of even distribution, optimization condition (\ref{opti relation}) has the simpler form
\eq
\label{opti even}
2tK\bar\tau\cos\left(2\alpha_0\sqrt{\bar\tau}\right)-2t^2\bar\tau\cos\left(2\alpha_0\sqrt{\bar\tau}\right)-z(K-2t)=0
\en
Subtracting equations (\ref{opti uneven}) and (\ref{opti even}), we get
\eq
\label{even equa uneven}
\sum_{i=1}^t2K\tau_i\cos\left(2\alpha_K\sqrt{\tau_i}\right)-2t\bar\tau\sum_{i=1}^t\cos\left(2\alpha_K\sqrt{\tau_i}\right)=2tK\bar\tau\cos\left(2\alpha_0\sqrt{\bar\tau}\right)-2t^2\bar\tau\cos\left(2\alpha_0\sqrt{\bar\tau}\right)
\en
Next, in second order of $\varepsilon_i$ (\ref{def varepsilon}), we expand the following expressions
\begin{multline}
\label{eqa cos}
\sum_{i=1}^t \cos\left(2\alpha_K\sqrt{\tau_i}\right)=t\cos\left(2\alpha_0\sqrt{\bar\tau}\right)-2t\sqrt{\bar\tau}\sin\left(2\alpha_0\sqrt{\bar\tau}\right)\Delta\alpha_K \\
+\left(\frac t {4\bar\tau\sqrt{\bar\tau}}\alpha_0\sin\left(2\alpha_0\sqrt{\bar\tau}\right)-\frac t {2\bar\tau}\alpha_0^2\cos\left(2\alpha_0\sqrt{\bar\tau}\right)\right)\delta^2(\tau);
\end{multline}
\vspace{-7mm}
\begin{multline}
\sum_{i=1}^t \tau_i\cos\left(2\alpha_K\sqrt{\tau_i}\right)=t\bar\tau\cos\left(2\alpha_0\sqrt{\bar\tau}\right)-2t\bar\tau\sqrt{\bar\tau}\sin\left(2\alpha_0\sqrt{\bar\tau}\right)\Delta\alpha_K\\
-\left(\frac {3t} {4\sqrt{\bar\tau}}\alpha_0\sin\left(2\alpha_0\sqrt{\bar\tau}\right)+\frac 1 2 t\alpha_0^2\cos\left(2\alpha_0\sqrt{\bar\tau}\right)\right)\delta^2(\tau)
\end{multline}
Substituting above relations into (\ref{even equa uneven}), we can solve $\Delta\alpha_K$ (\ref{def nabla}):
\eq
\label{tilde alpha}
\Delta\alpha_K=-\left(\frac{1-2\beta}{4\sqrt{3-4\beta}} \frac{\alpha_0^2}{\bar\tau\sqrt{\bar\tau}}+\frac{3+\beta}{8(1-\beta)}\frac{\alpha_0}{\bar\tau^2}\right)\delta^2(\tau),
\en
where $\beta=t/K$ and $\alpha_0$ is defined in (\ref{eta0 alpha0}).

  \item Analysing $\Delta\eta_K$ (\ref{def nabla}) by cancellation equation (\ref{eqa cancellation limit different}), we calculate $\Delta\eta_K$ explicitly in first order of $\delta^2(\tau)$ (\ref{def variance}). From cancellation equation (\ref{eqa cancellation limit different}), in first order of $\Delta\eta_K$,  we have
      \eq
      \label{eqa delta eta 1}
      \tan\left(2\eta_0\sqrt{\frac z K}\right)+2\sqrt{\frac{z}{K}}\left(1+\tan^2\left(2\eta_0\sqrt{\frac z K}\right)\right)\Delta\eta_K=\frac{2\sqrt K\sum_{i=1}^t\sqrt{\tau_i}\sin\left(2\alpha_K\sqrt{\tau_i}\right)}{\sqrt z\left(2\sum_{i=1}^t\cos\left(2\alpha_K\sqrt{\tau_i} \right)+K-2t\right)}
      \en
  For even distribution, we have the cancellation equation (\ref{eqa cancellation limit different}):
  \eq
  \label{eqa even cancellation}
  \tan\left(2\eta_0\sqrt{\frac{z}{K}}\right)=\frac{2t\sqrt{\bar\tau K}\sin\left(2\alpha_0\sqrt{\bar\tau}\right)}{\sqrt z\left(2t\cos\left(2\alpha_0\sqrt{\bar\tau}\right)+K-2t\right)}
  \en
  Substituting above equation and the approximate equations (\ref{eqa sin expand}) and (\ref{eqa cos}) in (\ref{eqa delta eta 1}), we can solve $\Delta\eta_K$ explicitly in first order of $\delta^2(\tau)$. The details of calculation please refer to the Appendix. The result is
  \eq
  \label{delta eta}
  \Delta\eta_{K}=-\left(
  \frac{(1-\beta)(1-2\beta)}{\sqrt{3-4\beta}}\frac{\alpha_0^2}{\bar\tau\sqrt{\bar\tau}}
  +\frac{(4\beta^3-8\beta^2+3\beta+1)}{4(1-\beta)^2}\frac{\alpha_0}{\bar\tau^2}
  +\frac{(1-2\beta)\sqrt{3-4\beta}}{16(1-\beta)}\frac 1 {\bar\tau^2\sqrt{\bar\tau}}
  \right)\delta^2(\tau)
  \en

  \item Prove the inequality (\ref{theorem ineqa}). The difference between number of queries for uneven distribution of target items and even distribution of target items is $f(\eta_0,\alpha_0)-f(\eta_K,\alpha_K)$, see (\ref{def total queries}). Because $f=\eta-\alpha$, we have
      \eq
      \label{f0 fk 1}
      f(\eta_0,\alpha_0)-f(\eta_K,\alpha_K)=\Delta\alpha_K-\Delta\eta_K
      \en
   According to $\Delta\alpha_K$ (\ref{tilde alpha}) and $\Delta\eta_K$ (\ref{delta eta}), we find
   \eq
   \Delta\alpha_K-\Delta\eta_k=\left(\frac{(1-2\beta){\sqrt{3-4\beta}}}{4}\frac{\alpha_0^2}{\bar\tau\sqrt{\bar\tau}}
  +\frac{\left(-8\beta^2+7\beta-1\right)}{8(1-\beta)}\frac{\alpha_0}{\bar\tau^2}+\frac{(1-2\beta)\sqrt{3-4\beta}}{16(1-\beta)}\frac{1}{\bar\tau^2\sqrt{\bar\tau}}        \right)\delta^2{(\tau)}
   \en
  In order to combine the three terms of above equation, we introduce the inequality
  \eq
  \label{f0 fk 2}
  \Delta\alpha_K-\Delta\eta_k>\left(\frac{\pi^2(1-2\beta){\sqrt{3-4\beta}}}{144}
  +\frac{\pi\left(-8\beta^2+7\beta-1\right)}{48(1-\beta)}+\frac{(1-2\beta)\sqrt{3-4\beta}}{16(1-\beta)}        \right)\frac{\delta^2{(\tau)}}{\bar\tau^{5/2}},
  \en
because $\alpha_0>\pi/6\sqrt{\bar\tau}$, see equation (\ref{eta0 alpha0}). At last, combine (\ref{f0 fk 1}) and (\ref{f0 fk 2}), then we have
  \eq
  f(\eta_0,\alpha_0)-f(\eta_K,\alpha_K)>g(\beta)\frac{\delta^2{(\tau)}}{\bar\tau^{5/2}}
  \en
  with
  \eq
  g(\beta)=\frac{\sqrt{3-4\beta}(1-2\beta)(\pi^2(1-\beta)+9)+3\pi(-8\beta^2+7\beta-1)}{144(1-\beta)}
  \en

  \item  The denominator of $g(\beta)$ is obviously positive. Therefore, we consider the inequality
      \eq
      \sqrt{3-4\beta}(1-2\beta)(\pi^2(1-\beta)+9)+3\pi(-8\beta^2+7\beta-1)>0,
      \en
      which has the solution $0<\beta<\beta_c$. Here $\beta_c$ is an algebraic number: it is a root of polynomial of $5$ degree:
      \eq
      \sqrt{3-4\beta_c}(1-2\beta_c)(\pi^2(1-\beta_c)+9)+3\pi(-8\beta_c^2+7\beta_c-1)=0
      \en
      Approximate expression is $\beta_c\approx 0.6281$. In many blocks limit ($K\rightarrow \infty$), namely $\beta=t/K\rightarrow 0$, we have
      \eq
      g(0)=\frac{\sqrt 3 \pi^2-3\pi+9\sqrt 3}{144}\approx 0.1615,
      \en
      which coincides with equation (\ref{eqa theorem 1}) in Theorem \ref{theorem 1}. The diagram of function $g(\beta)$ in the interval $(0,0.75)$ is in Figure \ref{fig1}. \textbf{So we proved Theorem \ref{theorem 2}}.

  \begin{figure}[htbp]
  \centering
  \includegraphics[width=9cm]{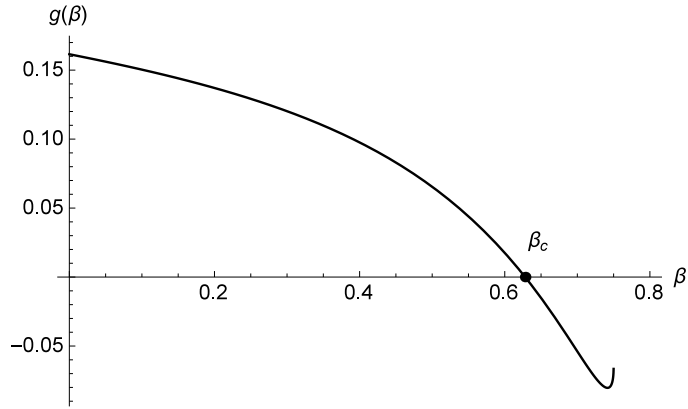}
  \caption{Diagram of function $g(\beta)$ (\ref{def g beta}) in region $\beta\in(0,0.75)$. Inequality $g(\beta)>0$ has the solution $0<\beta<\beta_c$.}\label{fig1}
  \end{figure}
  \end{enumerate}

\end{proof}

\section{Summary}

\label{section summary}

In this paper we study quantum partial search algorithm for multiple target items unevenly distributed in the target blocks (different target blocks have different number of target items). The number of global and local Grover iteration should satisfy the cancellation equation (\ref{eqa cancellation different}). We study the optimization of the algorithm and present the optimization condition (\ref{opti relation}) in large block limit ($b\rightarrow \infty$). We also prove that uneven distribution of multiple target items requires more queries than even distribution case by perturbation method for concentration of target blocks: $0<\beta<0.6281$ ($\beta=t/K$). It is open problem of studying the partial search algorithm for large $\beta$. Note that after the algorithm, the database will be in the state (\ref{eqa final state}).


\begin{appendices}

\section*{Appendix\hspace{.3cm} The proof of equation (\ref{delta eta})}

\label{appendix details}

The optimal value of $\alpha$ and $\eta$ (\ref{def eta alpha}) for even distribution of target items are denoted as $\eta_0$ and $\alpha_0$ (\ref{eta0 alpha0}). Their values can be rewritten as
\eq
\label{eta0 alpha0 beta}
\tan\left(2\eta_0\sqrt{\frac{z}{K}}\right)=\frac{\sqrt{3\beta-4\beta^2}}{1-2\beta},\quad\quad\quad \cos\left(2\alpha_0\sqrt{\bar\tau}\right)=\frac{1-2\beta}{2(1-\beta)},\quad\quad\quad \sin\left(2\alpha_0\sqrt{\bar\tau}\right)=\frac{\sqrt{3-4\beta}}{2(1-\beta)}
\en
with $\beta=t/K$. On the other hand, the approximate equations (\ref{eqa sin expand}) and (\ref{eqa cos}) can be reformulated as
\begin{align}
\label{eqa sin rewrite}
\sum_{i=1}^t\sqrt{\tau_i}\sin\left(2\alpha_K\sqrt{\tau_i}\right)=&t\sqrt{\bar\tau}\sin\left(2\alpha_0\sqrt{\bar\tau}\right)+P\delta^2(\tau);\\
\label{eqa cos rewrite}
\sum_{i=1}^t \cos\left(2\alpha_K\sqrt{\tau_i}\right)=&t\cos\left(2\alpha_0\sqrt{\bar\tau}\right)+Q\delta^2(\tau).
\end{align}
The coefficient $P$ is found as
\eq
P=-\frac 1 2 \sin\left(2\alpha_0\sqrt{\bar\tau}\right)\frac{t\alpha^2_0}{\sqrt{\bar\tau}}+\frac 1 4 \cos\left(2\alpha_0\sqrt{\bar\tau}\right)\frac{t\alpha_0}{\bar\tau}-\frac 1 8\sin\left(2\alpha_0\sqrt{\bar\tau}\right)\frac{t}{\bar\tau\sqrt{\bar\tau}}
+2t\bar\tau\cos\left(2\alpha_0\sqrt{\bar\tau}\right)\frac{\Delta\alpha_K}{\delta^2(\tau)}
\en
Substituting $\Delta\alpha_K$ (\ref{tilde alpha}) and the optimal value $\alpha_0$ (\ref{eta0 alpha0 beta}) into above relation, after some algebra, coefficient $P$ equals to
\eq
\label{eqa P}
P=-\frac{(1-\beta)}{\sqrt{3-4\beta}}\frac{t\alpha_0^2}{\sqrt{\bar\tau}}
-\frac{(1-2\beta)(1+\beta)}{4(1-\beta)^2}\frac{t\alpha_0}{\bar\tau}
+\frac{\sqrt{3-4\beta}}{16(1-\beta)}\frac{t}{\bar\tau\sqrt{\bar\tau}}
\en
Similar, the coefficient $Q$ in (\ref{eqa cos rewrite}) has the expression
\eq
\label{eqa Q}
Q=\frac{\sqrt{3-4\beta}}{2(1-\beta)^2}\frac{t\alpha_0}{\bar\tau\sqrt{\bar\tau}}
\en
With the help of approximation equations (\ref{eqa sin rewrite}) and (\ref{eqa cos rewrite}), the right hand side of equation (\ref{eqa delta eta 1}) in first order of $\delta^2(\tau)$ becomes
\begin{multline}
\text{RHS of (\ref{eqa delta eta 1})}=\frac{2\sqrt K \left(t\sqrt{\bar\tau}\sin\left(2\alpha_0\sqrt{\bar\tau}\right)+P\delta^2(\tau) \right)}
{\sqrt z \left(2t\cos\left(2\alpha_0\sqrt{\bar\tau}\right)+2Q\delta^2(\tau)+K-2t\right)}\\
=\frac{2t\sqrt{K\bar\tau}\sin\left(2\alpha_0\sqrt{\bar\tau}\right)}{\sqrt z\left(2t\cos\left(2\alpha_0\sqrt{\bar\tau}\right)+K-2t\right)}
+\frac{2\sqrt KP\delta^2(\tau)}{\sqrt z\left(2t\cos\left(2\alpha_0\sqrt{\bar\tau}\right)+K-2t\right)}
-\frac{4t\sqrt{K\bar\tau}\sin\left(2\alpha_0\sqrt{\bar\tau}\right) Q \delta^2(\tau)}
{\sqrt z\left(2t\cos\left(2\alpha_0\sqrt{\bar\tau}\right)+K-2t\right)^2}
\end{multline}
Combining above result with left hand side of equation (\ref{eqa delta eta 1}), and using the identity (\ref{eqa even cancellation}), we have the result
\eq
\left(1+\tan^2\left(2\eta_0\sqrt{\frac z K}\right)\right)\Delta\eta_K=
\frac{KP\delta^2(\tau)}{z\left(2t\cos\left(2\alpha_0\sqrt{\bar\tau}\right)+K-2t\right)}
-\frac{2tK\sqrt{\bar\tau}\sin\left(2\alpha_0\sqrt{\bar\tau}\right)A\delta^2(\tau)}{z\left(2t\cos\left(2\alpha_0\sqrt{\bar\tau}\right)+K-2t\right)^2}
\en
Note that $\eta_0$ and $\alpha_0$ have the expressions (\ref{eta0 alpha0 beta}). Then above equation gives
\eq
\Delta\eta_K
=\left((1-2\beta)\frac{P}{t\bar\tau}-\beta\sqrt{3-4\beta}\frac{Q}{t\sqrt{\bar\tau}}\right)\delta^2(\tau)
\en
Last, substituting coefficient $P$ (\ref{eqa P}) and $Q$ (\ref{eqa Q}) into above equation, after some algebra, we solve $\Delta\eta_K$ in first order of $\delta^2(\tau)$ explicitly:
\eq
  \Delta\eta_{K}=-\left(
  \frac{(1-\beta)(1-2\beta)}{\sqrt{3-4\beta}}\frac{\alpha_0^2}{\bar\tau\sqrt{\bar\tau}}
  +\frac{(4\beta^3-8\beta^2+3\beta+1)}{4(1-\beta)^2}\frac{\alpha_0}{\bar\tau^2}
  +\frac{(1-2\beta)\sqrt{3-4\beta}}{16(1-\beta)}\frac 1 {\bar\tau^2\sqrt{\bar\tau}}
  \right)\delta^2(\tau)
  \en

\end{appendices}



\begin{thebibliography}{}

   \bibitem{GR05} L.K. Grover and J. Radhakrishnan, {\it Is partial quantum search of a database any easier?}, Proceedings of the seventeenth annual ACM symposium on Parallelism in algorithms and architectures. ACM (2005).

  \bibitem{KG06} V.E. Korepin and L.K. Grover, {\it Simple algorithm for partial quantum search}, Quantum Inf. Process. {\bf 5}, 5-10 (2006).

  \bibitem{KX09} V.E. Korepin and Y. Xu, {\it Quantum search algorithms}, Int. J. Mod. Phys. B {\bf 23}, 5727-5758 (2009).

  \bibitem{Grover96} L.K. Grover, {\it A fast quantum mechanical algorithm for database search}, Proceedings of the twenty-eighth annual ACM symposium on Theory of computing, ACM (1996).

  \bibitem{Grover97} L.K. Grover, {\it Quantum mechanics helps in searching for a needle in a haystack}, Phys. Rev. Lett. {\bf 79}, 325 (1997).

 \bibitem{Zalka99} C. Zalka, {\it Grover¡¯s quantum searching algorithm is optimal}, Phys. Rev. A {\bf 60}, 2746 (1999).

 \bibitem{BBHT96} M. Boyer, G. Brassard, P. H{\o}yer and A. Tapp, {\it Tight bounds on quantum searching}, arXiv preprint quant-ph/9605034 (1996).

 \bibitem{Korepin05} V.E. Korepin, {\it Optimization of partial search}, J. Phys. A: Math. Gene. {\bf 38}, 731 (2005).

 \bibitem{Korepin06} V.E. Korepin and J. Liao, {\it Quest for fast partial search algorithm}, Quantum Inf. Process. {\bf 5}, 209-226 (2006).

 \bibitem{KX07} V.E. Korepin and Y. Xu, {\it Hierarchical quantum search}, Int. J. Mod. Phys. B {\bf 21}, 5187-5205 (2007).

 \bibitem{KV06} V.E. Korepin and B.C. Vallilo, {\it Group theoretical formulation of a quantum partial search algorithm}, Progress of Theoretical Physics {\bf 116}, 783-793 (2006).

 \bibitem{CV07}  B.-S. Choi and V.E. Korepin, {\it Quantum partial search of a database with several target items}, Quantum Inf. Process. {\bf 6},  243-254 (2007).

 \bibitem{CPB09} Pu-Cha Zhong, Wan-Su Bao and Yun Wei, {\it Quantum partial searching algorithm of a database with several target items}, Chin. Phys. Lett. {\bf 26}, 020301 (2009).

 \bibitem{NC11}  M.A. Nielsen and I.L. Chuang, {\it Quantum computation and quantum information}, Cambridge University Press, Cambridge, UK, 2000 and 2011.



 \end{thebibliography}


\end{document}